%% file: 7inputs.tex
\documentclass[preprint]{elsarticle}

\usepackage{amsmath}
\usepackage{amsthm}
\usepackage{amssymb}
\usepackage{color}
\usepackage{enumerate}

\newcommand{\alg}[1]{\texttt{#1}}

\newtheorem{lemma}{Lemma}
\newtheorem{corollary}{Corollary}
\newtheorem{theorem}{Theorem}
\newtheorem{definition}{Definition}
\newtheorem{example}{Example}

\begin{document}

\begin{frontmatter}

\title{When Six Gates are Not Enough\tnoteref{thanks}}
\tnotetext[thanks]{Supported by the Israel Science Foundation, grant 182/13, and by the Danish Council for Independent Research, Natural Sciences.}

\author[israel]{Michael Codish}
\author[denmark]{Lu\'{i}s Cruz-Filipe}
\author[israel]{Michael Frank}
\author[denmark]{Peter Schneider-Kamp}

\address[israel]{Department of Computer Science, Ben-Gurion University of the Negev, Israel}
\address[denmark]{Department of Mathematics and Computer Science, University of Southern Denmark}

\begin{abstract}
  We apply the pigeonhole principle to show that there must exist Boolean
  functions on $7$ inputs with a multiplicative complexity of at least $7$,
  i.e., that cannot be computed with only $6$ multiplications in the Galois field with two elements.
\end{abstract}


\end{frontmatter}

\section{Introduction}
The multiplicative complexity of a Boolean function is the minimal number of
multiplications over the Galois field $GF(2)$ needed to implement it.
As a measure of a function's non-linearity,
it is an important property with many applications, e.g., in the analysis of
cryptographic ciphers and hash functions \cite{Boyar2008}, or in the study of the communication
complexity of multiparty computation~\cite{Hirt2005}.

On a circuit level, multiplications over $GF(2)$ correspond to AND gates, while additions
correspond to XOR gates and the unit to the constant $\top$ (TRUE).
Thus, an equivalent characterization of the multiplicative complexity
of a Boolean function is to consider the minimal number of AND gates needed to implement the
function in the presence of an arbitrary number of XOR gates. It is this second characterization
which will be used throughout this paper.

Given a number of inputs $n$, the maximal multiplicative complexity of an $n$-ary Boolean function
is denoted by $M(n)$.
In other words, $M(n$) measures how much intrinsic non-linearity is possible given a fixed
number of arguments.
Determining lower bounds for $M(n)$ is an interesting question that has been widely addressed
e.g.~in~\cite{Boyar2008,Boyar2000,Turan2014}.
In this article, we apply a pigeonhole argument to prove that $M(7)\geq 7$, raising the previous best known lower bound by $1$.

The structure of this paper is as follows.
We present the necessary background in Section \ref{sec:background}, and define an abstract notion of topology of a circuit in Section~\ref{sec:def-topology}.
In Section \ref{sec:symmetry}, we introduce a symmetry break to reduce the upper bound on the number of Boolean functions of $n$ inputs computable by circuits with $k$ AND gates.
In Section \ref{sec:topology}, we study the different ways in which we can interconnect those AND gates, showing that we can drastically reduce the number of relevant circuits by a generate-and-prune algorithm inspired by~\cite{ourICTAIpaper}.
Combining these two results, we apply a pigeonhole counting argument in Section \ref{sec:result} to obtain our new lower bound.
We conclude with an outlook on future work in Section~\ref{sec:conclusion}.

\section{Background}
\label{sec:background}

A Boolean function on $n$ inputs, or an $n$-ary Boolean function, is a function from $\{0,1\}^n\to\{0,1\}$.
The set of all Boolean functions on $n$ inputs is denoted $B_n$, and $|B_n|=2^{2^n}$.
We will often write $\bot$ for $0$ and $\top$ for $1$.

It is well known that every Boolean function can be implemented by means of a circuit consisting of only AND ($\wedge$), XOR ($\oplus$) and NOT ($\neg$) gates.
Furthermore, since $\neg x=x\oplus\top$, the NOT gates can be removed if we allow the use of the constant $\top$.
As observed in~\cite{Boyar2000}, we can assume AND gates to be binary and XOR gates to have an unbounded number of inputs.
Such circuits are called XOR-AND circuits therein; in this paper, we will refer to them simply as \emph{circuits}.
Due to the associativity of XOR, any circuit with $k$ AND gates can therefore be specified using exactly $2k+1$ XOR gates: $2k$ of them producing the inputs for the AND gates, and an additional one to produce the output.

\begin{definition}
  \label{defn:circuit}
  For each natural number $n$, let $X_n=\{x_i\mid 1\leq i\leq n\}$ denote the $n$ inputs to a circuit, and $X_n^+=X_n\cup\{\top\}$.
  A circuit with $n$ inputs and $k$ AND gates is a pair $\mathcal C=\langle\mathcal A,\mathcal O\rangle$, where:
  \begin{itemize}
  \item $\mathcal A=\langle a_i\mid 1\leq i\leq k\rangle$ is a list of $k$ AND gates, where the $i$-th gate $a_i=\langle L_i,R_i\rangle$ with $L_i,R_i\subseteq\{a_j\mid 1\leq j<i\}\cup X^+$.
  \item $\mathcal O\subseteq\mathcal A\cup X_n^+$ is the output (XOR) gate.
  \end{itemize}
\end{definition}
Intuitively, each element of $\mathcal A$ represents an AND gate, whose inputs are the outputs of two XOR gates whose inputs are given by $L_i$ and $R_i$, which we will informally write as $\left(\bigoplus L_i\right)\wedge\left(\bigoplus R_i\right)$.
$\mathcal O$ represents the final XOR gate, and the function $f_{\mathcal C}$ computed by $\mathcal C$ returns the output from this gate.

\begin{example}
  \label{ex:circuit}
  Consider the circuit depicted in Figure~\ref{fig:circuit}, which computes the majority function on $4$ bits (returning $\top$ if at least three of the bits are $\top$).
  In our notation, this circuit is represented as $\mathcal C=\langle\mathcal A,\mathcal O\rangle$, where:
  \begin{align*}
    \mathcal A &= \langle a_1,a_2,a_3,a_4\rangle &
    a_1 &= \langle\{x_1\},\{x_2\}\rangle &
    a_3 &= \langle\{x_1,x_2\},\{a_2\}\rangle \\
    \mathcal O &= \{a_3,a_4\} &
    a_2 &= \langle\{x_3\},\{x_4\}\rangle &
    a_4 &= \langle\{a_1\},\{x_3,x_4,a_2\}\rangle
  \end{align*}
\end{example}

\begin{figure}
  \centering
  \scalebox{.71}{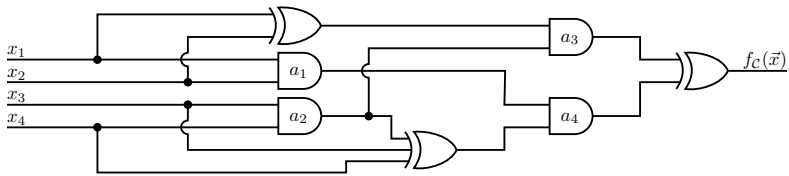}

  \caption{A circuit computing the majority function on $4$ bits.
    The labels on the AND gates are as in Example~\ref{ex:circuit}.
    Here, $f_{\mathcal C}(\vec x)=((x_1\oplus x_2)\wedge(x_3\wedge x_4))\oplus((x_1\wedge x_2)\wedge(x_3\oplus x_4\oplus(x_3\wedge x_4))$.}
  \label{fig:circuit}
\end{figure}

\begin{lemma}[Lemma~15 from~\cite{Boyar2000}]
  \label{lemma15}
  At most $2^{k^2+2k+2kn+n+1}$ functions from $B_n$ can be computed by circuits with $k$ AND gates.
\end{lemma}
\begin{proof}[Proof~\cite{Boyar2000}]
For the $i$-th gate, there are $2^{2(n+1+i-1)}$ possible sets $L_i$ and $R_i$: each may use the $n$ inputs, $\top$, and the $i-1$ previous AND gates.
For the output, there are $2^{n+1+k}$ possibilities. Thus, there are at most
$2^{n+1+k} \times \prod_{i=1}^k 2^{2(n+1+i-1)} = 2^{n+1+k+k(k+2n+1)} = 2^{k^2+2k+2kn+n+1}$
potentially computable functions.
\end{proof}
For $n = 7$ and $k = 6$, Lemma~\ref{lemma15} yields an upper bound of $2^{140}$ functions from $B_7$ computable with $6$ AND gates, i.e., $6$ AND gates are potentially enough to compute all $2^{2^7}$ Boolean functions with $7$ inputs.

Table~\ref{tab:M(n)} represents some known values and lower bounds for $M(n)$.
The fully-determined values of $M(n)$ for up to $4$ inputs are folklore, and easily shown to be correct, while $5$ was shown in \cite{Turan2014} using an exhaustive computer-based exploration of all $48$ equivalence classes of $B_5$.
The latter approach does not directly scale to $6$ inputs, as the number of equivalences classes of $B_6$ explodes to $150{,}357$.

The lower bound for $6$ inputs is based on the observation that trivially $M(n)\geq n-1$.
As the above table shows, this bound is tight for the determined values of $n\leq 5$.
The counting argument from~\cite{Boyar2000} gives a non-trivial lower bound for $n\geq 8$, leaving the open questions of whether the lower bounds for $6$ and $7$ inputs are tight.
We prove that this is not the case for $7$ inputs.

\begin{table}[b]
\centering
\begin{tabular}{l||c|c|c|c|c|c|c|c}
$n$ & $1$ & $2$ & $3$ & $4$ & $5$ & $6$ & $7$ & $8$\\
\hline
$M(n)$ & $0$ & $1$ & $2$ & $3$ & $4$ & $\geq 5$ & $\geq 6$ & $\geq 9$
\end{tabular}
\caption{Known determined values and lower bounds of $M(n)$ for up to $8$ inputs.}
\label{tab:M(n)}
\end{table}

\section{Topology of a circuit}
\label{sec:def-topology}

Our results capitalize on one abstraction: the notion of topology of a circuit,
which intuitively forgets all connections except those between the AND gates, distinguishing only the different ways in which they use each others' outputs.

\begin{definition}
  A \emph{(circuit) topology} is a set $\mathcal A$ of AND gates, as in Definition~\ref{defn:circuit}, except that $L\cup R\subseteq\mathcal A$ for all $\langle L,R\rangle\in\mathcal A$.
  Given an AND-XOR circuit $\mathcal C=\langle\mathcal A,\mathcal O\rangle$, the \emph{topology} of $\mathcal C$ is
  $\langle\langle L\cap\mathcal A,R\cap\mathcal A\rangle\mid\langle L,R\rangle\in\mathcal A\rangle$.
\end{definition}
Informally, a topology abstracts from the linear part of the circuit, considering only the connections between the AND gates; different circuits with the same topology can compute different Boolean functions.

\begin{example}
  The topology of the circuit $\mathcal C$ in Figure~\ref{fig:circuit} is
  $\{a_1,a_2,a_3,a_4\}$, with
  $a_1=a_2=\langle\emptyset,\emptyset\rangle$,
  $a_3=\langle\emptyset,\{a_2\}\rangle$ and
  $a_4=\langle\{a_1\},\{a_2\}\rangle$.
\end{example}

\begin{definition}
  Let $T$ be a topology.
  A function $f\in B_n$ is \emph{computable by $T$} if $f$ is computed by some circuit $\mathcal C$ whose topology is $T$.
\end{definition}

The notion of topology allows us to give a different proof of Lemma~\ref{lemma15}.
Since each AND gate consists of two subsets of the previous gates, the total number of different topologies on $k$ gates is
\begin{equation}
  \label{eq:topcount}
  \prod_{i=1}^k\left(2^{i-1}\right)^2 = 2^{\sum_{i=1}^k 2(i-1)} = 2^{k^2-k}\,.
\end{equation}
On the other hand, each input to each gate in a topology abstracts from $2^{n+1}$ concrete circuits (those containing the AND gates specified in the topology, plus any combination of circuit inputs and possibly $\top$), so there are
\begin{equation}
  \label{eq:circcount}
  \left(2^{n+1}\right)^{2k}\times 2^{k+n+1}
\end{equation}
circuits with any given topology, where the second term in this product counts the number of possibilities for the output gate.
Combining both estimates, we obtain a total of
$2^{k^2-k}\times\left(2^{n+1}\right)^{2k}\times 2^{k+n+1} = 2^{k^2-k+2kn+2k+k+n+1} = 2^{k^2+2k+2kn+n+1}$
different circuits.
In the next sections, we will optimize the bounds in Equations~\eqref{eq:topcount} and~\eqref{eq:circcount} separately.

\section{Breaking symmetry on negations}
\label{sec:symmetry}

In this section, we note that there are different circuits with the same number of AND gates that compute the same $n$-ary Boolean functions, and that we can provide a syntactic characterization for many of these, thus improving the bound of Equation~\eqref{eq:circcount}.

\begin{definition}
  Let $\mathcal C=\langle\mathcal A,\mathcal O\rangle$ be a circuit.
  We say that $\mathcal C$ is \emph{negation-normal} if there is no gate $\langle L,R\rangle\in\mathcal A$ such that $\top\in L\cap R$.
\end{definition}

\begin{lemma}
  \label{lem:not-not}
  Every $n$-ary Boolean function computable by a circuit with $k$ AND gates can be computed by a negation-normal circuit with $k$ AND gates.
\end{lemma}
\begin{proof}
  By using the equivalence
  $(X\oplus\top)\wedge(Y\oplus\top)\equiv(X\wedge Y)\oplus X\oplus Y\oplus\top$
  we can rewrite any circuit so that no AND gate has $\top$ added to both its inputs.
  Observe that both sides of the equation use only one AND gate.
\end{proof}

\begin{theorem}
  \label{teo:functions}
  The number of negation-normal circuits on $n$ inputs with a given topology on $k$ AND gates is at most
  $\left(3\times 2^{2n}\right)^k\times2^{n+k+1}$.
\end{theorem}
\begin{proof}
  The argument is similar to the one establishing Equation~\eqref{eq:circcount}.
  Each AND gate in the topology corresponds to $3\times 2^n\times 2^n$ possibilities: each input can receive any subset of circuit inputs (the two $2^n$ factors), and either one may also receive $\top$, but not both.
  The possibilities for the output gate are unchanged.
\end{proof}

Combining this result with Equation~\eqref{eq:topcount}, we obtain the following result.
\begin{corollary}
  At most $3^k\times 2^{k^2+2kn+n+1}$ functions from $B_n$ can be computed by circuits with $k$ AND gates.
\end{corollary}
On its own, this (small) improvement does not produce any new lower bounds for $M(n)$; in particular, for $n=7$, the number of functions potentially computable with $6$ AND gates becomes $3^6 \times 2^{36+84+7+1} > 2^9 \times 2^{128} = 2^{137}$.

\section{Breaking symmetry on topologies}
\label{sec:topology}

We now focus on improving the bound in Equation~\eqref{eq:topcount} by showing that some topologies compute the same functions.

\begin{definition}
  The set $\mathcal T^0_k$ is the set of all possible topologies with $k$ AND gates.
\end{definition}

Our goal is to remove elements from $T^0_k$ while preserving the set
of all functions computable by a topology in that set.
The first observation is that the actual order of the AND gates is irrelevant for the function computed by the actual circuit, so we can eliminate topologies that only differ on these labels.

\begin{definition}
  Two topologies $T$ and $T'$ are \emph{equivalent}, denoted
  $T\equiv T'$, if there is a permutation $\pi$ of $\{1,\ldots,n\}$
  such that: $\langle L,R\rangle\in T$ iff either
  $\langle\pi(L),\pi(R)\rangle\in T'$ or
  $\langle\pi(R),\pi(L)\rangle\in T'$, where $\pi$ is
  structurally extended to sets and pairs.
\end{definition}
It is easy to check that this relation is an equivalence relation.

\begin{lemma}
  \label{lem:equiv}
  Let $T$ and $T'$ be topologies, with $T\equiv T'$, and $\mathcal C$ be a circuit with topology $T$.
  Then there is a circuit $\mathcal C'$ with topology $T'$ such that $f_{\mathcal C}=f_{\mathcal C'}$.
\end{lemma}
\begin{proof}
  Construct $\mathcal C'$ by renaming the AND gates in $\mathcal C$ according to $\pi$.
  By commutativity and associativity of $\oplus$, together with commutativity of $\wedge$, a straightforward reasoning by induction establishes that $f_{\mathcal C}^{a_i}=f_{\mathcal C'}^{a_{\pi(i)}}$ for $1\leq i\leq k$, and therefore that $f_{\mathcal C}=f_{\mathcal C}^{\mathcal O}=f_{\mathcal C'}^{\pi(\mathcal O)}=f_{\mathcal C'}$.
\end{proof}

Consecutive AND gates in a topology can be grouped in disjoint \emph{layers},
such that the gates in each layer only depend on the outputs of gates
in previous layers.
The algorithm in Figure~\ref{fig:layering} computes the maximal layering of the gates --
the one such that no layer can be extended forward.

\begin{figure}[hb]
  \centering
  \fbox{\begin{minipage}{.5\textwidth}\small
    \begin{tabbing}
    {\bf(output)} \= \qquad \= \qquad \= \kill
    Algorithm \alg{Layering}\\ \\
    {\bf(input)} \> topology
    $T=\langle\langle L_i,R_i\rangle\mid 1\leq i\leq k\rangle$ \\
    {\bf(init)} \> $\ell:=1$, $S_1:=\emptyset$ \\
    {\bf(loop)} \> for $i=1..k$ \+\+\\
    if $S_\ell\cap(L_i\cup R_i) = \emptyset$ \\
    then $S_\ell := S_\ell\cup\{a_i\}$ \\
    else $\ell := \ell+1$, $S_\ell = \{a_i\}$ \-\-\\
    {\bf(output)} \> layering $S_1,\ldots,S_\ell$
  \end{tabbing}
  \end{minipage}}

  \caption{Algorithm \alg{Layering} to compute a maximal layering of a topology.}
  \label{fig:layering}
\end{figure}

The following definition captures the idea that gates should only be in a layer $\ell$ if one of their inputs depends on a gate in the previous layer $\ell-1$.

\begin{definition}
  A topology $T=\langle\langle L_i,R_i\rangle\mid i=1,\ldots,n\rangle$ is
  \emph{well-layered} if its layering $S_1,\ldots,S_m$ is such that, for
  every $i$ and $k$, if $a_i\in S_k$, then
  $L_i\cap S_{k-1}\neq\emptyset$.
\end{definition}

\begin{example}
  The topology from the circuit in Figure~\ref{fig:circuit} has layers $\{a_1,a_2\}$ and $\{a_3,a_4\}$, and thus it is well-layered, as both $a_3$ and $a_4$ use the output of $a_2$.

  The topology $\{a'_1,a'_2,a'_3,a'_4\}$ for the same circuit, where $a'_1=a_2$, $a'_2=a_3$, $a'_3=a_1$ and $a'_4=a_4$, is not well-layered: its layers are $\{a'_1\}$, $\{a'_2,a'_3\}$ and $\{a'_4\}$, and gate $a'_3$ does not use any gate in the previous layer.
\end{example}

\begin{lemma}[Layering]
  \label{lem:normalization}
  Every topology is equivalent to a well-layered topology.
\end{lemma}
\begin{proof}
  Let $T=\langle\langle L_i,R_i\rangle\mid i=1,\ldots,k\rangle$ and
  $S_1,\ldots,S_m$ be its layering.
  Assume $T$ is not well-layered, and let $i$ be the smallest index such
  that $a_i\in S_\ell$ and \mbox{$L_i\cap S_{\ell-1}=\emptyset$}.
  
  If $R_i\cap S_{\ell-1}$, then build $T'$ by replacing
  $\langle L_i,R_i\rangle$ with $\langle R_i,L_i\rangle$ in $T$.
  Otherwise, let $j=\max\{z\mid a_z\in L_i\cup R_i\}$, with $\max(\emptyset)=0$;
  let $\pi$ be the permutation inserting $i$ between $j$ and $j+1$
  (so $\pi(i)=j+1$, $\pi(z)=z+1$ for $j<z<i$, and $\pi(z)=z$ for all
  other $z$), and take $T'=\pi(T)$, interchanging $L_i$ and $R_i$ in
  $a_i$ if $a_j\in R_i$.
  Observe that $T'$ is still a valid topology.

  In either case, all indices up to $i$ satisfy the layering
  condition.
  In the first case this is trivial; in the second case, note that
  $j$ cannot occur in $L_{j+1},\ldots,L_i$ or $R_{j+1},\ldots,R_i$ in
  $T'$, so $j+1,\ldots,i$ remain in the same layers as the corresponding
  $j,\ldots,i-1$ in the layering of $T$.

  Iterating this
  construction yields a well-layered topology equivalent to $T$.
\end{proof}

\begin{corollary}
  Let $\mathcal T^1_k$ be the set of well-layered topologies in $\mathcal T^0_k$.
  If $f\in B_n$ is computable by a topology in $\mathcal T^0_k$, then it is computable by a topology in $\mathcal T^1_k$.
\end{corollary}
\begin{proof}
  Consequence of Lemmas~\ref{lem:equiv} and~\ref{lem:normalization}.
\end{proof}

We now begin to eliminate redundant topologies from $\mathcal T^1_n$.
Our results make use of the following identity, valid for all Boolean values $P$ and $Q$.
\begin{equation}
  \label{lem:and-plus}
  P\wedge Q \equiv P\wedge(P\oplus Q\oplus\top)
\end{equation}

\begin{definition}
  A topology $T$ is \emph{minimal} if the following hold
  for all $\langle L,R\rangle\in T$.
  \begin{enumerate}[(i)]
  \item (A) If $L\neq\emptyset$, then $L\not\subseteq R$, and
    (B) If $R\neq\emptyset$, then $R\not\subseteq L$.
  \item If $L\cap R\neq\emptyset$, then $(L\cap R)<L\setminus R$ and
    $(L\cap R)<R\setminus L$, where $<$ is any (fixed) total ordering of
    $\wp(\{a_1,\ldots,a_k\})$.
  \end{enumerate}
\end{definition}

\begin{lemma}
  \label{lem:minimal}
  If $f\in B_n$ is computable by topology $T$,
  then it is computable by a well-layered and minimal topology $T'$ with the same number of AND gates as $T$.
\end{lemma}
\begin{proof}
  Let $\mathcal C$ be a circuit computing $f$ with topology $T$.
  Without loss of generality we can assume $T$ is well-layered.
  Assume also that $T$ is not minimal.
  We show that we can transform $\mathcal C$ so that the three conditions are
  met; at each stage, the triple $\langle v_1,v_2,v_3\rangle$
  indicating the number of gates violating conditions~(i-A), (i-B) and
  (ii), respectively, decreases w.r.t.~lexicographic ordering.
  Since $\mathcal C$ is finite, iteration produces a circuit with minimal
  topology.
  \begin{enumerate}[(i)]
  \item[(i-A)] Assume that gate $a=\langle L,R\rangle$ is such that $L\subseteq R$, so that
    $R=L\cup R'$.
    Then the function computed by this gate can be written as
    $((\bigoplus L)\oplus A)\wedge((\bigoplus L)\oplus(\bigoplus R')\oplus B)$,
    and by~\eqref{lem:and-plus} this is equivalent to
    $((\bigoplus L)\oplus A)\wedge((\bigoplus R')\oplus A\oplus B\oplus\top)$.
    Replacing $a$ by $\langle L,R'\rangle$ yields a circuit that has
    one less violation of condition~(i-A).
  \item[(i-B)] Assume that gate $a=\langle L,R\rangle$ is such that $R\subseteq L$, so that
    $L=L'\cup R$.
    The construction is analogous, using the equivalence between
    $((\bigoplus L')\oplus(\bigoplus R)\oplus A)\wedge((\bigoplus R)\oplus B)$
    and $((\bigoplus L')\oplus A\oplus B\oplus\top)\wedge((\bigoplus R)\oplus B)$.

    In order to ensure that the resulting topology is well-layered, it might
    be necessary to interchange $L'$ and $R$ in the gate replacing $a$,
    as possibly only $R$ intersects the previous layer.

  \item[(ii)] Assume that gate $a=\langle L,R\rangle$ is such that $L\cap R\neq\emptyset$, so
    that $L=X\cup L'$ and $R=X\cup R'$, with all of $L'$, $R'$ and
    $X$ not empty (otherwise condition (i) would not be met).
    Again by~\eqref{lem:and-plus} we can write the function
    computed by this gate as one of
    \begin{align*}
      \textstyle((\bigoplus X)\oplus(\bigoplus L')\oplus A)
      &\textstyle{}\bigwedge{}((\bigoplus X)\oplus(\bigoplus R')\oplus B) \\
      \textstyle((\bigoplus X)\oplus(\bigoplus L')\oplus A)
      &\textstyle{}\bigwedge{}((\bigoplus L')\oplus(\bigoplus R')\oplus A\oplus B\oplus T) \\
      \textstyle((\bigoplus L')\oplus(\bigoplus R')\oplus A\oplus B\oplus T)
      &\textstyle{}\bigwedge{}((\bigoplus X)\oplus(\bigoplus R')\oplus B)
    \end{align*}
    and we can replace $a$ by a gate whose inputs intersect on either $X$, $L'$ or $R'$,
    which means we can always ensure it to be the lexicographically
    smallest of the three.

    Since either $X$ or $L'$ intersects the previous layer, it is also
    possible to guarantee layering, if necessary by permuting the
    inputs.
    Likewise, the resulting gate always satisfies condition~(i).\qedhere
  \end{enumerate}
\end{proof}

\begin{definition}
  The set $\mathcal T_k$ is the set of all well-layered and minimal
  topologies using $k$ AND gates.
\end{definition}

Merging Lemmas~\ref{lem:normalization} and~\ref{lem:minimal}, we obtain the following result.

\begin{theorem}
  \label{teo:Tk-complete}
  Every $n$-ary Boolean function computable by a circuit with $k$ AND gates is computable by a topology in $\mathcal T_k$.
\end{theorem}

The iterative algorithm in Figure~\ref{fig:g-and-p} computes a
set of minimal, well-layered topologies unique up to equivalence -- in other
words, representatives of the elements of $\mathcal T_k/{\equiv}$.
It generates these topologies layer by layer, pruning those equivalent to some other, in the spirit of~\cite{ourICTAIpaper}.
In the last line of the \textbf{(loop)} in \alg{Extend}, the notation $T\cdot a$ denotes the list obtained by appending gate $a$ to $T$.

\begin{figure}[ht]
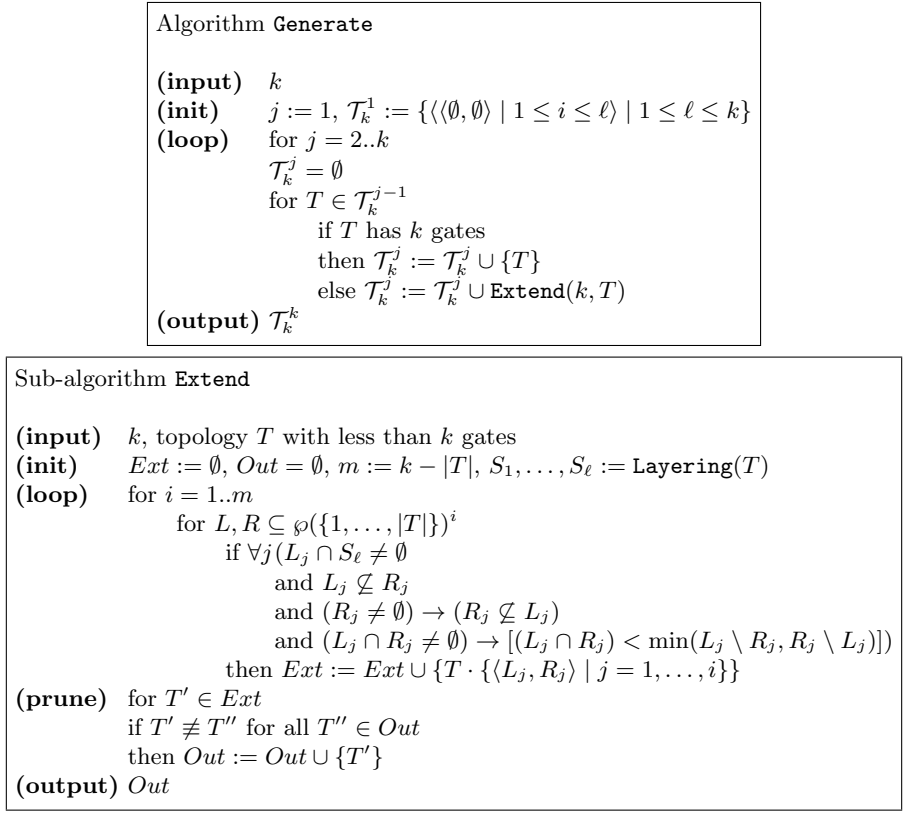

  \centering

  \fbox{\begin{minipage}{.7\textwidth}\small
  \begin{tabbing}
    {\bf(output)} \= \qquad \= \qquad \= \qquad \= \kill
    Algorithm \alg{Generate}\\ \\
    {\bf(input)} \> $k$ \\
    {\bf(init)} \> $j:=1$, $\mathcal T^1_k:=\{\langle\langle \emptyset,\emptyset\rangle\mid1\leq i\leq \ell\rangle\mid 1\leq \ell\leq k\}$ \\
    {\bf(loop)} \> for $j=2..k$ \+ \\
    $\mathcal T^j_k = \emptyset$ \\
    for $T\in\mathcal T^{j-1}_k$ \+ \\
    if $T$ has $k$ gates \\
    then $\mathcal T^j_k := \mathcal T^j_k\cup\{T\}$ \\
    else $\mathcal T^j_k := \mathcal T^j_k\cup\mbox{\alg{Extend}}(k,T)$ \-\- \\
    {\bf(output)} \> $\mathcal T^k_k$
  \end{tabbing}
  \end{minipage}}\smallskip

  \fbox{\begin{minipage}{.7\textwidth}\small
  \begin{tabbing}
    {\bf(output)} \= \qquad \= \qquad \= \qquad \= \kill
    Sub-algorithm \alg{Extend} \\ \\
    {\bf(input)} \> $k$, topology $T$ with less than $k$ gates \\
    {\bf(init)} \> $Ext:=\emptyset$, $Out=\emptyset$,
    $m:=k-|T|$, $S_1,\ldots,S_\ell:=\mbox{\alg{Layering}}(T)$ \\
    {\bf(loop)} \> for $i=1..m$ \+\+ \\
    for $L,R\subseteq\wp(\{1,\ldots,|T|\})^i$ \+\\
    if $\forall j$ \> $(L_j\cap S_\ell\neq\emptyset$ \+ \\
    and $L_j\not\subseteq R_j$ \\
    and $(R_j\neq\emptyset)\to(R_j\not\subseteq L_j)$ \\
    and $(L_j\cap R_j\neq\emptyset)\to[(L_j\cap R_j)<\min(L_j \setminus R_j,R_j \setminus L_j)])$ \-\\
    then $Ext := Ext\cup\{T\cdot\{\langle L_j,R_j\rangle\mid j=1,\ldots,i\}\}$ \-\-\-\\
    {\bf(prune)} \> for $T'\in Ext$ \+ \\
    if $T'\not\equiv T''$ for all $T''\in Out$ \\
    then $Out := Out\cup\{T'\}$ \- \\
    {\bf(output)} \> $Out$
  \end{tabbing}
  \end{minipage}}

  \caption{Iterative algorithm \alg{Generate} to compute $\mathcal T_k/{\equiv}$.}
  \label{fig:g-and-p}
\end{figure}

\begin{theorem}
  If $T\in\mathcal T_k$, then $T\equiv T'$ for some
  $T'\in\mbox{\alg{Generate}}(n)$.
\end{theorem}
\begin{proof}
  A topology with $k$ gates has at most $k$ layers, and \alg{Generate}
  loops through all possible lengths of these layers.

  In \alg{Extend}, we loop over all possible combinations of outputs
  from previous gates.
  The condition in the innermost loop excludes gates that lead to
  non-well-layered or non-minimal topologies.
  The pruning step guarantees that the first representative of each
  equivalence class of topologies is kept.

  Therefore every minimal and well-layered topology is equivalent to an
  element of \alg{Generate}$(k)$.
\end{proof}

Table~\ref{tab:topologies} shows the sizes of the sets
$\mathcal T_k/\equiv$, computed using two independent implementations of Algorithm \alg{Generate}.

\begin{table}[h]
  \centering
  \begin{tabular}{r|r|r|r|r|r|r}
    $k$ & $1$ & $2$ & $3$ & \multicolumn1{c|}{$4$} & \multicolumn1{c|}{$5$} & \multicolumn1c{$6$} \\ \hline
    $\left|\mathcal T_k/\equiv\right|$ & $1$ & $2$ & $8$ & $88$ & $3{,}564$ & $555{,}709$ \\
  \end{tabular}
  \caption{Number of non-equivalent minimal well-layered topologies using $k$ AND gates.}
  \label{tab:topologies}
\end{table}

Replacing the estimated number of topologies on $k$ AND gates given in Equation~\eqref{eq:topcount} reduces the straightforward upper bound on the number of computable functions on $7$ inputs with $6$ AND gates from $2^{140}$ to $555{,}709 \times 2^{110} > 2^{19} \times 2^{110} = 2^{129}$, which is still (just) larger than the number of $7$-ary Boolean functions.
However, combining this result with Theorem~\ref{teo:functions} does produce a new result, presented in the next section.

\section{The result}
\label{sec:result}

Combining Theorems~\ref{teo:functions} and~\ref{teo:Tk-complete} we immediately obtain the following result.

\begin{theorem}
  \label{teo:upbound}
  At most $3^k\times 2^{2kn+n+k+1}\times\left|\mathcal T_k/\equiv\right|$ functions from $B_n$ can be computed by circuits with $k$ AND gates.
\end{theorem}

\begin{theorem}
  There is a Boolean function on $7$ inputs with a multiplicative complexity of $7$ or higher.
\end{theorem}
\begin{proof}
  By Table~\ref{tab:topologies}, there are $555{,}709$ possible
  topologies for circuits with $6$ AND gates.
  Instantiating $n=7$ and $k=6$ in Theorem~\ref{teo:upbound} and using this value,
  we conclude that the number of $7$-ary Boolean functions computable by circuits
  with $6$ gates is at most
  $555{,}709\times3^6\times2^{98}<2^{20}\times2^{10}\times2^{98}=2^{128}=|B_7|$.
  Therefore, not all functions in $B_7$ can be computable by these circuits.
\end{proof}

\section{Conclusion and Future Work}
\label{sec:conclusion}
In this work we have shown that $M(7)$ is at least $7$, raising the previously known lower bound by $1$. The case of $7$
inputs has consequently become the smallest known case where $M(n) > n-1$.

In the future, we are planning to determine $M(6)$, which we conjecture to be~$5$, by extensive computer experiments refining the approach of \cite{Turan2014}.
Also, we plan to find an actual Boolean function on $7$ inputs with a multiplicative complexity of $7$ or higher as a witness
to our non-constructive proof.

\bibliographystyle{elsarticle-num}
\bibliography{7inputs}

\end{document}

%% file: circuit.pdf_tex
\begingroup%
  \makeatletter%
  \providecommand\color[2][]{%
    \errmessage{(Inkscape) Color is used for the text in Inkscape, but the package 'color.sty' is not loaded}%
    \renewcommand\color[2][]{}%
  }%
  \providecommand\transparent[1]{%
    \errmessage{(Inkscape) Transparency is used (non-zero) for the text in Inkscape, but the package 'transparent.sty' is not loaded}%
    \renewcommand\transparent[1]{}%
  }%
  \providecommand\rotatebox[2]{#2}%
  \ifx\svgwidth\undefined%
    \setlength{\unitlength}{415bp}%
    \ifx\svgscale\undefined%
      \relax%
    \else%
      \setlength{\unitlength}{\unitlength * \real{\svgscale}}%
    \fi%
  \else%
    \setlength{\unitlength}{\svgwidth}%
  \fi%
  \global\let\svgwidth\undefined%
  \global\let\svgscale\undefined%
  \makeatother%
  \begin{picture}(1,0.21445748)%
    \put(0,0){\includegraphics[width=\unitlength]{circuit.pdf}}%
    \put(0.00120481,0.15253011){\color[rgb]{0,0,0}\makebox(0,0)[lb]{\smash{$x_1$}}}%
    \put(0.00120481,0.12361445){\color[rgb]{0,0,0}\makebox(0,0)[lb]{\smash{$x_2$}}}%
    \put(0.00120481,0.09469879){\color[rgb]{0,0,0}\makebox(0,0)[lb]{\smash{$x_3$}}}%
    \put(0.00120481,0.06578312){\color[rgb]{0,0,0}\makebox(0,0)[lb]{\smash{$x_4$}}}%
    \put(0.3626506,0.12554218){\color[rgb]{0,0,0}\makebox(0,0)[lb]{\smash{$a_1$}}}%
    \put(0.3626506,0.06771085){\color[rgb]{0,0,0}\makebox(0,0)[lb]{\smash{$a_2$}}}%
    \put(0.99879517,0.13903614){\color[rgb]{0,0,0}\makebox(0,0)[rb]{\smash{$f_{\mathcal C}(\vec x)$}}}%
    \put(0.70963855,0.16891567){\color[rgb]{0,0,0}\makebox(0,0)[lb]{\smash{$a_3$}}}%
    \put(0.70963855,0.06771085){\color[rgb]{0,0,0}\makebox(0,0)[lb]{\smash{$a_4$}}}%
  \end{picture}%
\endgroup%

%% file: 7inputs.bbl
\begin{thebibliography}{1}
\expandafter\ifx\csname url\endcsname\relax
  \def\url#1{\texttt{#1}}\fi
\expandafter\ifx\csname urlprefix\endcsname\relax\def\urlprefix{URL }\fi
\expandafter\ifx\csname href\endcsname\relax
  \def\href#1#2{#2} \def\path#1{#1}\fi

\bibitem{Boyar2008}
J.~Boyar, R.~Peralta, Tight bounds for the multiplicative complexity of
  symmetric functions, Theor. Comput. Sci. 396~(1--3) (2008) 223--246.

\bibitem{Hirt2005}
M.~Hirt, J.~B. Nielsen, Upper bounds on the communication complexity of
  optimally resilient cryptographic multiparty computation, in: B.~K. Roy
  (Ed.), {ASIACRYPT} 2005, Vol. 3788 of LNCS, Springer, 2005, pp. 79--99.

\bibitem{Boyar2000}
J.~Boyar, R.~Peralta, D.~Pochuev, On the multiplicative complexity of boolean
  functions over the basis ($\wedge$, $+$, $1$), Theor. Comput. Sci. 235~(1)
  (2000) 43--57.

\bibitem{Turan2014}
M.~S. Turan, R.~Peralta, The multiplicative complexity of boolean functions on
  four and five variables, in: T.~Eisenbarth, E.~{\"{O}}zt{\"{u}}rk (Eds.),
  LightSec 2014, Vol. 8898 of LNCS, Springer, 2015, pp. 21--33.

\bibitem{ourICTAIpaper}
M.~Codish, L.~Cruz-Filipe, M.~Frank, P.~Schneider-Kamp, Twenty-five comparators
  is optimal when sorting nine inputs (and twenty-nine for ten), in: ICTAI
  2014, IEEE, 2014, pp. 186--193.

\end{thebibliography}
